\documentclass[12pt]{article}

\usepackage[margin=1in]{geometry}

\usepackage{amsmath,amssymb,mathtools}
\usepackage{booktabs}
\usepackage{microtype}
\usepackage{authblk}
\usepackage{mathrsfs}
\usepackage{appendix}

\usepackage{graphicx}
\usepackage{float}
\usepackage{caption}
\usepackage[section]{placeins}
\usepackage{lscape}
\usepackage{booktabs}
\usepackage{threeparttable}

\usepackage{xurl}

\usepackage[
  backend=biber,
  bibstyle=authoryear,
  citestyle=authoryear-comp,
  maxcitenames=2,
  mincitenames=1,
  maxbibnames=99,
  minbibnames=99,
  sortcites=true,
  date=year,
  dashed=false,
  giveninits=true,
  uniquename=false,
  uniquelist=false,
  doi=true,
  url=true,
  isbn=false
]{biblatex}
\usepackage[hidelinks]{hyperref}

\DeclareDelimFormat[parencite]{nameyeardelim}{\addcomma\space}

\DeclareNameAlias{author}{family-given}
\DeclareNameAlias{editor}{family-given}
\DeclareNameAlias{translator}{family-given}

\DeclareDelimFormat{multinamedelim}{\addcomma\space}
\DeclareDelimFormat{finalnamedelim}{\addspace\bibstring{and}\space}

\DeclareFieldFormat[article]{title}{#1}
\DeclareFieldFormat[incollection]{title}{#1}

\renewbibmacro*{in:}{}

\DeclareFieldFormat[article]{volume}{#1}
\DeclareFieldFormat[article]{number}{\mkbibparens{#1}}
\renewbibmacro*{volume+number+eid}{%
  \printfield{volume}%
  \printfield{number}%
  \setunit{\bibeidpunct}%
  \printfield{eid}}

\renewbibmacro*{journal+issuetitle}{%
  \usebibmacro{journal}%
  \setunit*{\addcomma\space}%
  \iffieldundef{series}
    {}
    {\newunit
     \printfield{series}%
     \setunit{\addspace}}%
  \usebibmacro{volume+number+eid}%
  \setunit{\addspace}%
  \usebibmacro{issue+date}%
  \setunit{\addcolon\space}%
  \usebibmacro{issue}%
  \newunit}

\DeclareFieldFormat{pages}{#1}

\AtEveryBibitem{%
  \ifentrytype{book}
    {\clearfield{volume}}
    {}}

\DeclareFieldFormat[online]{url}{\url{#1}}
\renewbibmacro*{url+urldate}{\usebibmacro{url}}

\renewbibmacro*{url+urldate}{\usebibmacro{url}}

\usepackage{amsthm}

\theoremstyle{plain}
\newtheorem{Proposition}{Proposition}

\theoremstyle{remark}
\newtheorem{Assumption}{Assumption}

\newtheoremstyle{boldremark}{}{}{\normalfont}{}{\bfseries \itshape}{.}{ }{}
\theoremstyle{boldremark}
\newtheorem{Estimand}{Estimand}
\def\E{\mathrm{E}}

\def\ATME{\mathit{ATME}}
\def\AMTME{\mathit{AMTME}}

\def\pages{\mathit{pages}}
\def\candidates{\mathit{candidates}}
\def\win{\mathit{win}}

\def\indep{\perp\!\!\!\perp}
\newcommand{\percent}{{\fontfamily{ppl}\selectfont\%}}

\allowdisplaybreaks[3]

\title{Candidate Set Size and Voting Behavior: A Front-Door Approach to Causal Moderation}
\author{Masayuki Haruhara}
\affil{Hitotsubashi University, 2-1 Naka, Kunitachi, Tokyo 186-8601, Japan \\ \href{mailto: em265012@g.hit-u.ac.jp}{em265012@g.hit-u.ac.jp}}
\date{\empty}

\begin{document}
    \maketitle
    \begin{abstract}
    We study causal moderation when treatment assignment is randomized but the moderator is not. We combine the parallel estimation framework with front-door adjustment to identify an average mediated treatment moderation effect. We apply this approach to 165 municipal assembly elections in Tokyo (1987–2023), where pamphlet positions are assigned by lottery and total pamphlet pages are mechanically determined by candidate set size and fixed municipal rules. One additional candidate increases the front-page effect on winning probability by 0.322 percentage points through the page-count channel, 8.0\percent{} of its baseline magnitude, while the back-page effect does not increase significantly.
\end{abstract}


\clearpage
\section{Introduction}
    The number of candidates in elections has been increasing worldwide. For example, the average number of candidates in U.S. House primary elections in 2020 was approximately 1.4 times that in 2010 \parencite{goidel2025choice}. A larger candidate set can, in principle, make it more likely that voters find a candidate who is close to their preferences \parencite{cunow2021less}.\par
However, a larger number of candidates also increases the number of comparisons voters must make to decide whom to vote for, thereby increasing the burden on voters in the voting decision-making process. In such cases, decision makers may adopt simpler decision-making processes that do not examine all available information, or may abandon the decision itself \parencite{iyengar2000choice, payne1993adaptive}. Thus, an increase in the number of candidates may also create a gap between voters' true preferences and their actual voting behavior.\par
When insufficient information causes actual choices not to reflect decision makers' true preferences, their welfare may be reduced. It is therefore important to analyze how an increase in the number of candidates affects voters' voting decision making. Such a phenomenon has already been observed, albeit in limited settings. For example, \textcite{cunow2021less} and \textcite{soderlund2021coping} report that, as the number of candidates increases, the tendency for candidates listed near the top of the ballot to receive more votes (ballot order effect: BOE) becomes stronger. However, existing studies analyze only decisions made over short periods in laboratories or at polling places. By contrast, many voters in actual elections are expected to make their voting decisions over the course of campaign periods lasting several days or longer. It is therefore not entirely appropriate to assess the desirability of increasing the number of candidates solely on the basis of changes in voting decision making under severe time constraints. Accordingly, this paper aims to quantify whether an increase in the number of candidates changes voting decision making in settings with relatively weak time constraints.\par
For this purpose, this paper uses the Candidate Information Pamphlet (CIP), which is used in Japanese election campaigns, to capture changes in voting decision making under relatively weak time constraints. The CIP is a document containing candidates' names, personal histories, policy views, etc. (Public Offices Election Act, Article 167) and is a multipage document distributed to households by election management committees, including through newspaper inserts, during the campaign period. \textcite{haruhara2026unintended} exploit the fact that the pages on which candidates appear in the CIP are determined by lottery and show that candidates appearing on the front page (the first page) or back page (the last page) have higher winning probabilities and vote shares (front / back page effect). They further note that the CIP is available for viewing at home for several days or longer and interpret differences in winning probabilities and vote shares across CIP pages as evidence that voters adopt simpler voting decision-making processes even when time constraints are relatively weak. Following their interpretation, this paper treats changes in the front / back page effect as changes in voters' voting decision-making processes. On this basis, using data from municipal assembly elections held in Tokyo, Japan, from 1987 to 2023, this paper examines whether an increase in the number of candidates changes the front / back page effect.\par
In the empirical analysis, we first note that the front / back page effect is a treatment effect---the effect of being listed on the front / back page, rather than on inner pages, on winning probabilities and vote shares. We therefore model the change in the front / back page effect induced by an increase in the number of candidates as the interaction effect between the number of candidates and front / back page placement on the winning probability. In addition, exploiting the fact that candidates' page assignments in the CIP are determined by lottery, we apply the parallel estimation framework (PEF) of \textcite{bansak2021estimating}. In their framework, when the treatment is randomized, the effect of a moderator on the outcome is identified within each stratum defined by the treatment, and the difference across strata identifies how the moderator changes the treatment effect, that is, the interaction effect between the two variables. Because the treatment in this paper has three categories---front, back, and inner page(s)---we construct a strategy for identifying the effect of the number of candidates on the winning probability within each stratum.\par
Page placement in the CIP is randomized, whereas the number of candidates, the second variable, is not. Thus, even if the number of candidates is correlated with the winning probability\footnote{The winning probability in this context is aggregated so that it can be compared with the number of candidates, which is an election-level variable. In general, the average winning probability across all candidates in an election is the number of seats divided by the number of candidates, whereas the average winning probability within strata defined by front-page, back-page, and inner-page placement is not mechanically determined.}, this correlation cannot immediately be interpreted as causal. For example, heightened electoral salience may both increase the number of candidates and induce voters to examine candidate information more carefully, thereby preventing an increase in the winning probabilities of candidates listed on the front / back page. To address such confounding, this paper uses the front-door criterion (FDC) of \textcite{pearl1995causal}. The FDC identifies the causal effect of an explanatory variable on an outcome by using an exogenous mediator between the explanatory variable and the outcome. In this paper, we posit that an increase in the number of candidates affects each candidate's winning probability through changes in the total number of CIP pages\footnote{An increase in the number of candidates likely affects each candidate's winning probability through paths other than the one mediated by the total number of CIP pages, but the method used in this paper cannot identify the effects along such paths. The purpose of the analysis, however, is to capture the relationship between the number of candidates and changes in winning probabilities attributable specifically to the front / back page effect. For this purpose, focusing only on the path mediated by the total number of CIP pages is more appropriate.}. The total number of CIP pages is determined solely by the ordinance of each municipality, which specifies the number of candidate slots per page, and by the number of candidates. Moreover, the ordinances specifying the number of candidate slots per page do not change during each municipality's sample period. Variation in the total number of CIP pages therefore arises only from the number of candidates and differences in ordinances across municipalities. Accordingly, by including municipality fixed effects in the model, the total number of CIP pages serves as an exogenous mediator.\par
The results show that one additional candidate increases the front page effect by 0.322 percentage points through an increase in the total number of CIP pages. This increase is substantial, amounting to 8.0\percent{} of the baseline front page effect of 4.023 percentage points. Thus, an increase in the number of candidates changes voting decision making even when time constraints are relatively weak.\par
This paper extends the conclusions of the literature on the relationship between the number of candidates and voters' voting decision making to a more realistic setting and to the decision making of a broader set of voters. Many studies have shown that, in settings with severe time constraints, such as experimental environments and polling places, an increase in the number of candidates leads voters to behave consistent with adopting simpler voting decision-making processes \parencite[e.g.,][]{cunow2021less, goidel2025choice, meredith2013causes, soderlund2021coping}. The results of this paper show that this simplification also occurs during the information-acquisition process in the campaign period, when time constraints are relatively weak.\par
In addition, this paper contributes to applied econometrics by providing a way to identify FDC settings that are more suitable than those in existing studies \parencite[e.g.,][]{bellemare2024paper, glynn2017front, glynn2018front, pearl1995causal}. The setting in this paper exploits an institutional rule (algorithm) that remains unchanged and has conditional branches depending only on the main explanatory variable, thereby generating an exogenous mediator. The FDC is considered difficult to apply because suitable settings are hard to find \parencite{huntington2021effect}; the insight from the setting identified in this paper mitigates this difficulty.\par
Finally, this paper extends applications of \textcite{bansak2021estimating} by combining the PEF with an explicit causal identification strategy for the non-randomized moderator. Several experimental studies apply their PEF when the treatment is randomized \parencite[e.g.,][]{carlsson2024politicians, kozyreva2023resolving, krajewski2025people, pickett2024officer}. In these studies, however, identification of the causal effect of the moderator on the outcome goes no further than adjustment for observed covariates and therefore cannot be said to fully resolve endogeneity due to unobserved factors. \textcite{bernardi2025socioeconomic} and \textcite{emeriau2026welcome} also apply the PEF in observational studies, but the treatment is not randomized in \textcite{bernardi2025socioeconomic}, while \textcite{emeriau2026welcome} adjust only for observed covariates to address the endogeneity of the moderator. By contrast, in this paper, the treatment is randomized and the conditions under which the FDC identifies causal effects are satisfied. This paper therefore provides a model case for applying \textcite{bansak2021estimating}.

\section{Related Literature}

    \subsection{The Decision-making Process under Excessive Choices}
        Voters' decision making in elections consists of choosing one option from a set comprising all candidates and abstention. The process includes deciding what information to acquire and how much information to acquire about the election as a whole and about all candidates, and then determining, on the basis of the information obtained and prior information, which option in the choice set best fits their preferences. A large body of research has pointed out that such decision-making processes and their outcomes may differ depending on whether the number of available options is small or large.\par
Under a benchmark in which decision makers can costlessly evaluate all available options, expanding the choice set cannot reduce the maximum attainable utility. However, with limited cognitive resources, larger choice sets may induce decision makers to simplify their decision-making processes \parencite{payne1993adaptive}. If such simplification occurs before option attributes are recognized, decision makers may disregard important attributes that distinguish among options or exclude from consideration options whose attributes differ substantially from those of other options. In such cases, an increase in the number of options may reduce the range of option attributes considered by the decision maker, so a monotonically increasing relationship between the number of options and the decision maker's utility need not hold.\par
Building on these foundational theoretical arguments, experimental studies in psychology have found that, as the number of options increases, individuals become more likely to forgo a final choice and, even when they make a choice, tend to evaluate the chosen option less favorably. \textcite{iyengar2000choice} compare choice sets consisting of 6 versus 24 varieties of jam and 6 versus 30 varieties of chocolate, and show that larger choice sets reduce the probability of purchase and lower post-choice satisfaction. Similar patterns have also been documented for pens \parencite{shah2007buying}, gift boxes \parencite{reutskaja2009satisfaction}, and 401(k) plans \parencite{sethi2004much}. This phenomenon is known as choice overload and has been attributed, among other mechanisms, to higher cognitive costs of identifying and evaluating differences among alternatives \parencite{keller1987effects} and greater regret over hasty choices \parencite{inbar2011decision}.\par
Related to these findings in psychology, economics has developed several models suited to decision making when many options are available and has tested them using experimental and observational data. One prominent framework assumes that decision makers consider only a subset of the presented choice set and make their final choice from this consideration set \parencite{cattaneo2020random, manzini2014stochastic}. In addition to approaches that parametrically specify rules of attention, methods have been proposed to identify the formation of consideration sets nonparametrically under weak assumptions \parencite{cattaneo2020random}. Studies focusing on the search process also use satisficing rules under which search stops once an option satisfying an aspiration level / reservation utility is reached \parencite{caplin2011search, simon1955behavioral}, and some studies attempt to identify search and stopping rules using eye-tracking data \parencite{reutskaja2011search}. Rational inattention models, which endogenize the allocation of attention under information-processing costs \parencite{mackowiak2023rational, matvejka2015rational, sims2003implications}, have also been used to address closely related questions.

    \subsection{The Voter Behavior under Excessive Number of Candidates}
        As with increases in the number of options in other choice problems, an increase in the number of candidates may raise the cost of acquiring and comparing all information and encourage the use of simpler decision-making processes. \textcite{aguilar2015ballot} conduct a survey experiment on the streets of Brazil and report that subjects are more likely to vote for a candidate of the same self-reported race when presented with 12 candidates rather than 3. Race is among the candidate attributes that voters can readily observe. Greater reliance on such information is therefore consistent with a simplification of the decision-making process through the omission of information that is relatively difficult to acquire.\par
Furthermore, \textcite{cunow2021less} and \textcite{goidel2025choice} experimentally show that, as the number of candidates increases, the tendency for candidates listed near the beginning of the ballot to receive more votes (the BOE) becomes stronger. In actual elections, \textcite{meredith2013causes} and \textcite{soderlund2021coping} similarly document that the BOE becomes stronger as the number of candidates increases. In related work, using Japanese election data, \textcite{haruhara2026unintended} show that the tendency for candidates listed on the front page of an official medium providing candidate information (the CIP) to receive more votes (the front page effect) becomes stronger as the number of candidates increases. The BOE and the front page effect are consistent with decision-making processes in which voters stop searching before reaching the end of an ordering of candidates. These findings therefore suggest that an increase in the number of candidates promotes the use of alternative decision-making processes.\par
However, these findings may not be sufficient to conclude that the decision-making processes of most voters in actual elections change with the number of candidates. First, because the BOE concerns the effect of candidates' order on the ballot on vote shares, its analysis does not reflect the behavior of voters who have completed their voting decisions before arriving at the polling place. Consequently, evidence that the BOE becomes stronger as the number of candidates increases does not necessarily imply that the voting decisions of many voters, which are made over sufficient time during the campaign period, change as the number of candidates increases.\par
In addition, \textcite{lau2001advantages} point out that actual elections contain many features that are difficult to reproduce experimentally, so voting decision-making processes and their outcomes in experimental settings may differ from those in actual elections. For example, when an incumbent seeks reelection, voters can use abundant information about the incumbent's performance in office without acquiring additional information. Moreover, because election campaigns generally last several days or longer, voters can collect information and decide how to vote with sufficient time and while discussing the election with people around them. These conditions are difficult to reproduce in experimental settings. The external validity of findings from experimental studies such as \textcite{aguilar2015ballot} and \textcite{cunow2021less} may therefore be limited.\par
Another literature examines the relationship between the number of electoral options and abstention, but its findings are mixed. Some studies report that an increase in the number of candidates increases intentional invalid voting \parencite{nagler2015trading}, whereas others report that it increases turnout \parencite{bol2022does, goidel2025choice, pons2018expressive}. Other studies find that the relationships between the number of candidates or parties and invalid-vote rates or turnout vary with the electoral system \parencite{boulding2015political, cohen2018dynamic}, or that these relationships are inverted U-shaped \parencite{taagepera2014turnout}. Thus, there is no unified conclusion as to whether an increase in the number of candidates changes the voting decision-making process in a way that increases the abandonment of choice.

\section{Data}
    This paper uses data from municipal assembly elections held in municipalities\footnote{Municipalities in Tokyo consist of 23 special wards, 26 cities, and others.} in Tokyo, Japan, from 1987 to 2023. The dataset contains data on 8,450 candidates from 165 elections\footnote{Relatively few records are preserved for elections held before 2000, which account for only 20\percent{} (33 elections) of this dataset. Similarly, records for city assembly elections are less frequently preserved than those for ward assembly elections and account for only 19.4\percent{} (32 elections) of this dataset. The main focus of the analysis is therefore ward assembly elections held in Tokyo's 23 special wards from 2000 to 2023 (101 elections, or 61\percent{} of the dataset). This sample is almost identical to the ward assembly elections in Tokyo's 23 special wards from 2005 to 2023 analyzed by \textcite{haruhara2026unintended}.}, which we collected from the 487 elections held in relavant area and period. These data were collected manually from Records of Elections archived at the National Diet Library and municipal archives, Election.com, which provides election information from across Japan, and other sources.\footnote{Candidate gender ceased to be an officially disclosed item for elections held from July 2020 onward and therefore became unavailable from official sources. We consequently collected gender information from Election.com; the official websites of candidates and their affiliated parties; X (formerly Twitter); Facebook; and other publicly available sources.}

    \subsection{Candidate Information Pamphlet}
        The CIP analyzed in this paper is a document containing candidates' names, personal histories, policy views, etc. (Public Offices Election Act, Article 167). As shown in Figure 1, candidate information in a CIP is placed within a grid of boxes, while the margins mainly contain voting information. The number of rows and columns in this grid (the pamphlet format), such as four rows by three columns in the document on the left of Figure 1 and five rows by three columns in the document on the right, is constant within an election. A CIP therefore consists of multiple pages with the same number of rows and columns, with information on the candidate assigned to each box. Figure 2 lists the pamphlet formats used in the elections in the dataset, and Figure 3 shows the number of observations for each pamphlet format. Figure 3 shows that the 3 $\times$ 2, 4 $\times$ 3, and 5 $\times$ 3 pamphlet formats are used most frequently.
\begin{figure}[htbp]
    \centering
    \begin{minipage}{\textwidth}
        \centering
        \caption{Example of a Candidate Information Pamphlet}
        \includegraphics[width=\linewidth]{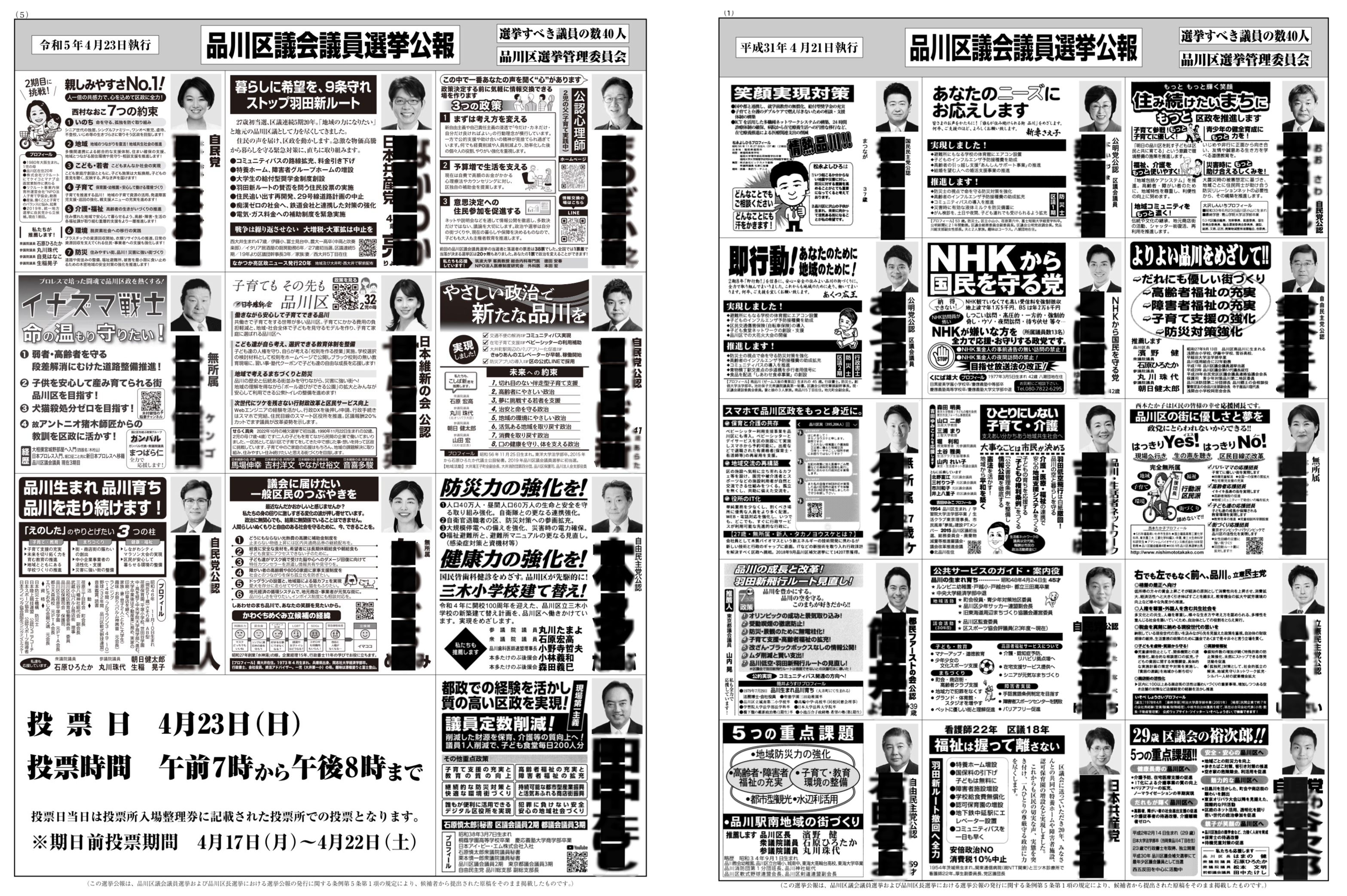}
        \begin{flushleft}
            \footnotesize
            \textbf{Note}: The sections displaying candidates' names have been modified in the image.\\
            \textbf{Source}: Provided by Shinagawa City Election Management Committee; reproduced with permission.
        \end{flushleft}
    \end{minipage}
\end{figure}
\par
\begin{figure}[htbp]
    \centering
    \begin{minipage}{\textwidth}
        \centering
        \caption{Pamphlet formats list}
        \includegraphics[width=\linewidth]{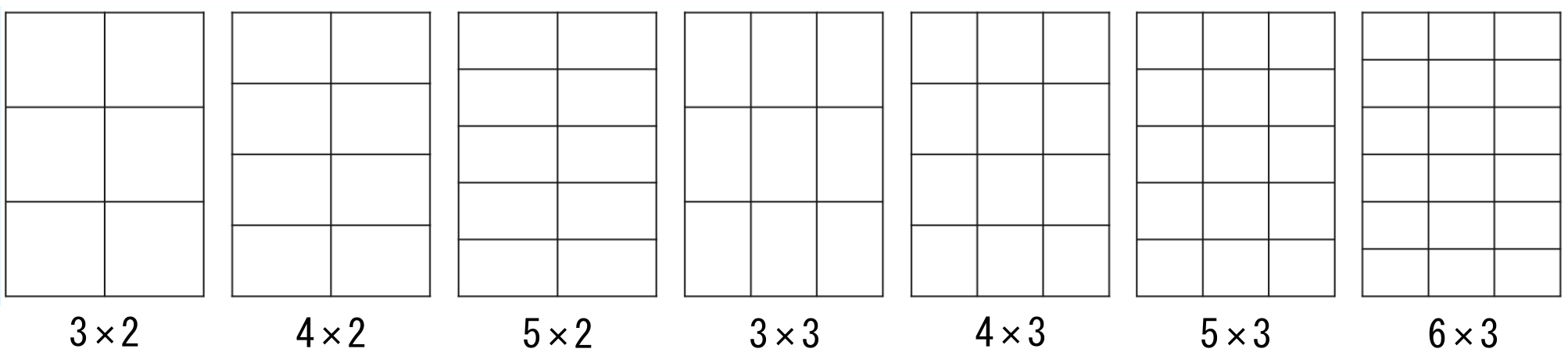}
        \begin{flushleft}
            \footnotesize
            \textbf{Note}: Pamphlet formats are denoted as ``rows × columns'' (e.g., 3 × 2 indicates three rows and two columns per page).
        \end{flushleft}
    \end{minipage}
\end{figure}
\par
\begin{figure}[htbp]
    \centering
    \begin{minipage}{0.75\textwidth}
        \centering
        \caption{Number of observations for each pamphlet format}
        \includegraphics[width=\linewidth]{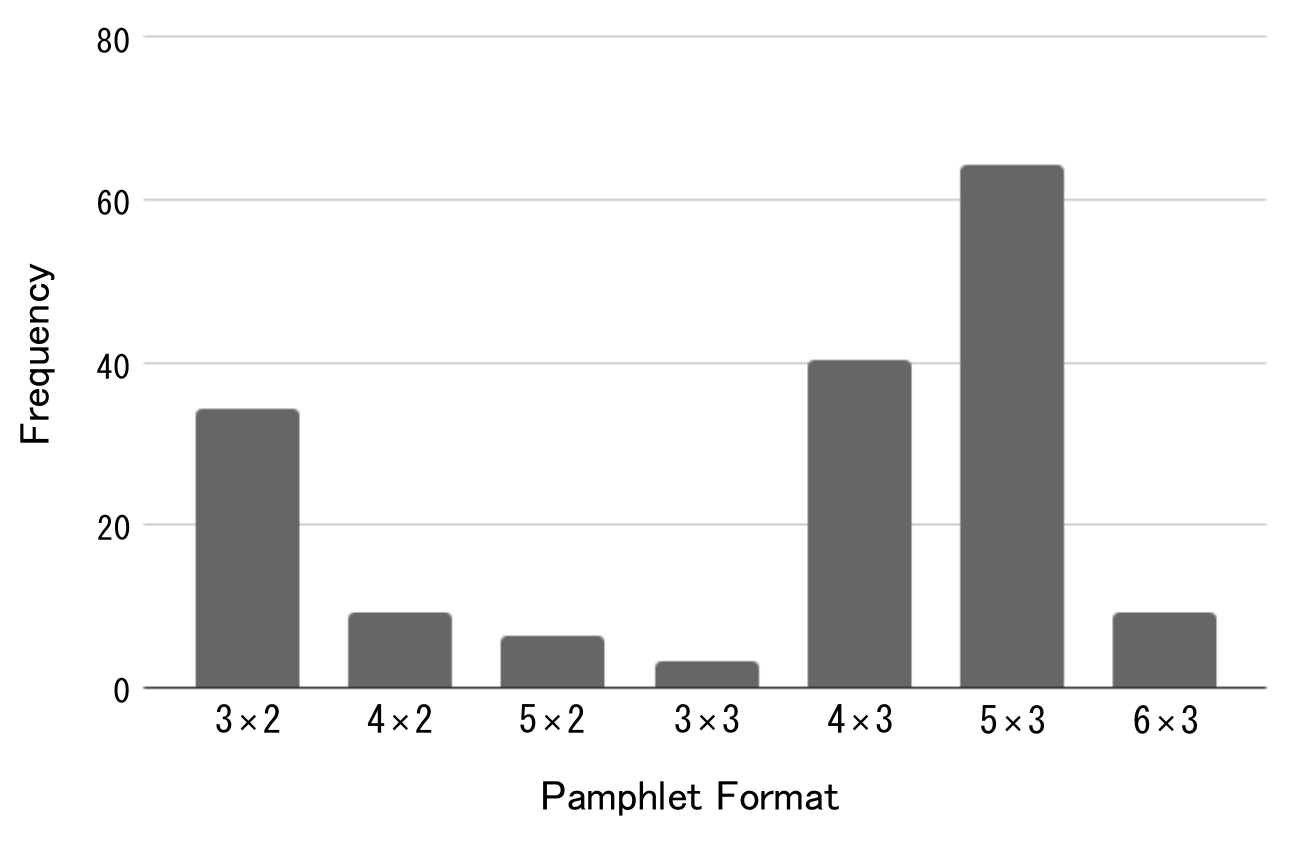}
        \begin{flushleft}
            \footnotesize
            \textbf{Note}: Pamphlet formats are denoted as ``rows×columns'' (e.g., 3 × 2 indicates three rows and two columns per page).
        \end{flushleft}
    \end{minipage}
\end{figure}
\par
The pamphlet format for each election is specified by the ordinance of the municipality, conditional on the number of candidates. For example, Article 82 of the Ota City Election Administration Regulations stipulates that when there are two or fewer candidates submitting materials for publication, the pamphlet shall have one page with one row and two columns; when there are three, one page with one row and three columns; when there are four, two pages with one row and two columns; when there are five to eight, two pages with two rows and two columns; when there are nine to twelve, one page with four rows and three columns; when there are thirteen to twenty-four, two pages with four rows and three columns; and thereafter two pages shall be added for every additional twenty-four candidates. Ordinances specifying the pamphlet format may be revised by decisions of the municipal assembly.\footnote{For example, Setagaya City revised its ordinance specifying the pamphlet format in 2011.} We therefore identified, for each municipality, a period containing no revision to the ordinance specifying the pamphlet format and collected only elections held during that period for the dataset.\footnote{When an ordinance was revised, we include the period before or after the revision that yields the larger number of elections for which data can be collected (the election-level sample size).} Consequently, the pamphlet formats used in the elections analyzed in this paper are unaffected by ordinance revisions and vary only with differences in municipality-specific ordinances and the number of candidates in each election.\par
The pamphlet format uniquely determines the number of candidate slots per page. In addition, a CIP has no front or back cover; candidate information and voting guidance are printed on every page. The total number of CIP pages is therefore determined by the number of candidate slots per page and the number of candidates. Accordingly, variation in the total number of CIP pages, which serves as the mediator in this paper, arises only from differences in municipality-specific ordinances and the number of candidates in each election. Figure 4 shows the distribution of the number of candidate slots per page, a scatter plot of the number of candidates against the number of slots per page, and a scatter plot of the number of slots per page against the total number of pages. Figure 4 suggests that the number of candidates is positively correlated with the number of slots per page, while the number of slots per page is negatively correlated with the total number of pages. Thus, although an increase in the number of candidates increases the total number of pages holding the number of slots per page fixed, the unconditional correlation between the number of candidates and the total number of pages may be weak or negative. Indeed, as shown in Figure 5, the relationship between the number of candidates and the total number of pages is not sufficiently clear to characterize as positively correlated. One possible explanation is that municipalities set the number of slots per page so that the total number of pages does not increase substantially as the number of candidates rises, thereby limiting the cost of issuing the CIP.
\begin{figure}[htbp]
    \centering
    \begin{minipage}{\textwidth}
        \centering
        \caption{Distribution of the number of candidate slots per page and scatter plots with the number of candidates and the total number of pages}
        \includegraphics[width=\textwidth]{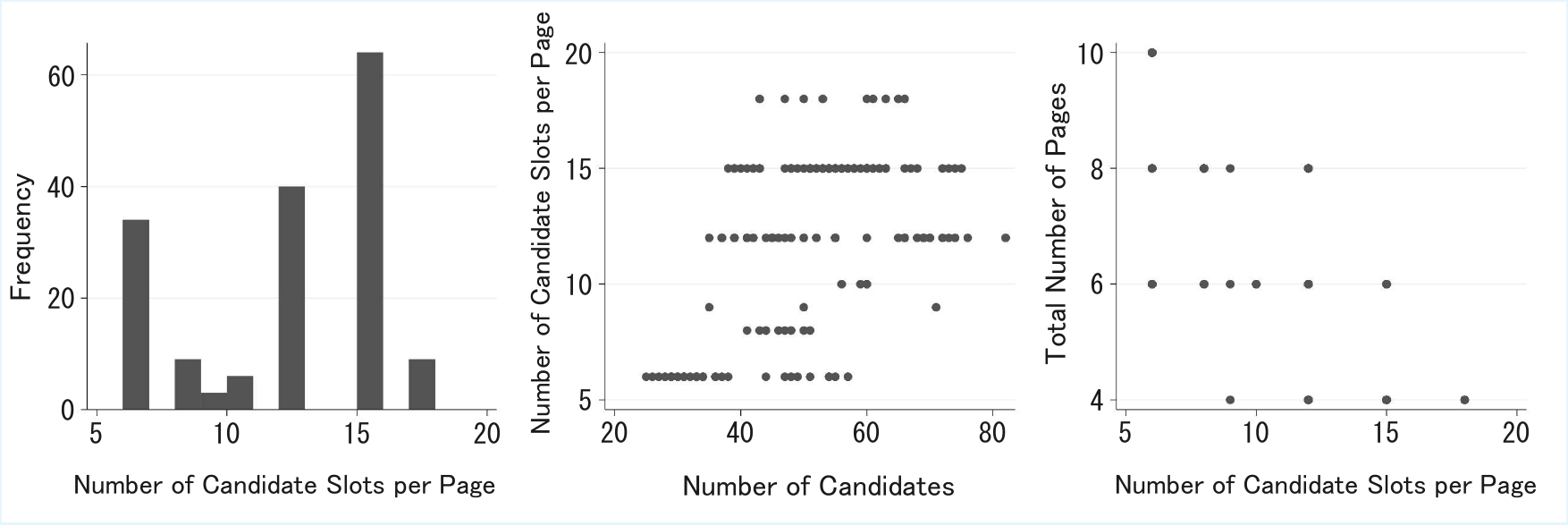}
    \end{minipage}
\end{figure}
\par
\begin{figure}[htbp]
    \centering
    \begin{minipage}{\textwidth}
        \centering
        \caption{Distribution of the number of candidates and the total number of pages, and their scatter plot}
        \includegraphics[width=\textwidth]{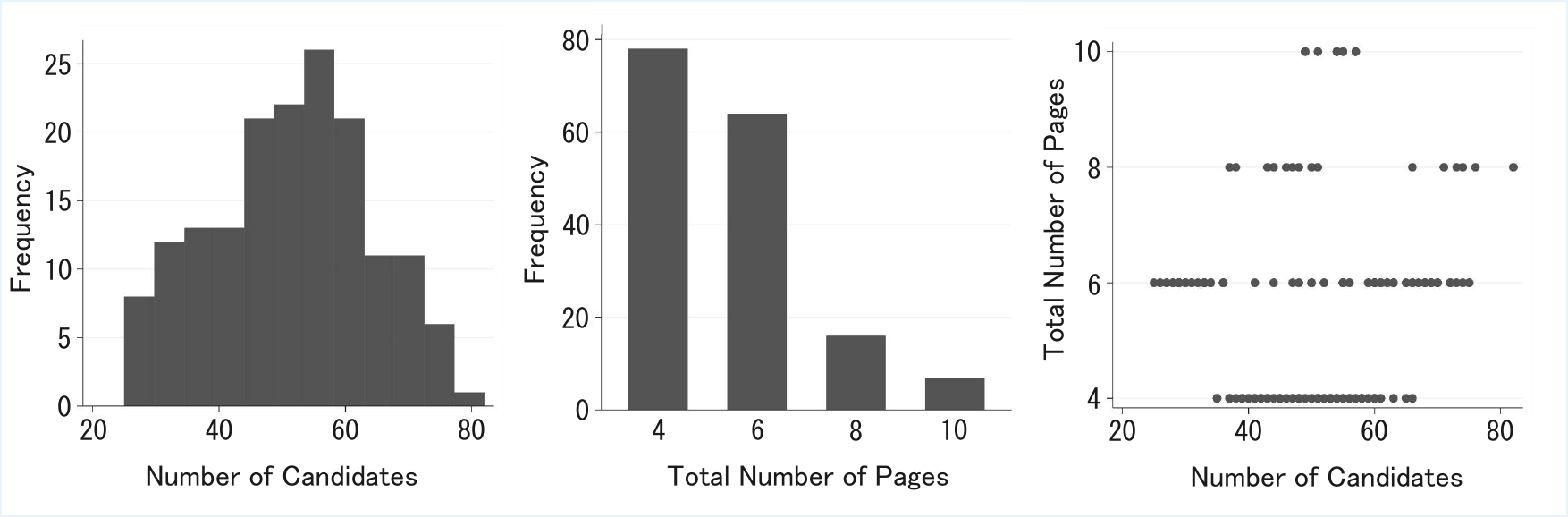}
    \end{minipage}
\end{figure}
\par
Moreover, in every election in the dataset, municipal ordinances require candidates' positions in the CIP to be determined by lottery (see \textcite{haruhara2026unintended} for details). Thus, whether each candidate is listed on the front page (the first page), back page (the last page), or inner pages (all other pages) is randomly determined.\par
The CIP is issued by each municipality's election management committee and distributed to voters' households by the day before the election, generally including through newspaper inserts. Voters can therefore spend time reading candidate information and deciding how to vote during a campaign period that generally lasts several days or longer. Surveys have repeatedly confirmed that many voters in Japanese municipal assembly elections actually follow such a voting decision-making process. In opinion surveys conducted by \textcite{fairelections2012, fairelections2016, fairelections2019, fairelections2024}, respondents were allowed to select multiple campaign activities, including posters and street speeches, that they considered ``useful.'' The CIP was the item selected by the largest number of respondents.

    \subsection{Summary Statistics}
        Tables 1 and 2 report summary statistics for the variables used in this paper. At the election level, the dataset contains the number of candidates as the main explanatory variable and the total number of CIP pages as the mediator. At the candidate level, it contains a winning dummy as the outcome, page placement in the CIP as the main explanatory variable, and age, gender, candidate status (incumbent versus newcomer / former), and party affiliation as covariates. In addition, the dataset was collected and constructed to include only elections with no missing values for any of these variables.
\begin{table}[H]
    \centering
    \begin{threeparttable}
        \caption{Summary statistics for election-level variables}
        \begin{tabular}{@{}lccccc@{}}
            \toprule
            & Observations & Mean & SD & Min & Max \\
            \midrule
            Number of candidates & 165 & 51.212 & 12.483 & 25 & 82\\
            Total number of pages & 165 & 5.418 & 1.623 & 4 & 10\\
            Number of candidate slots per page & 165 & 11.909 & 3.759 & 6 & 18\\
            \bottomrule
        \end{tabular}
    \end{threeparttable}
\end{table}
\par
The main object of analysis in this paper is the effect of the number of candidates on the front / back page effect. The front / back page effect refers to the causal effect of front / back page placement on the winning probability: the extent to which a candidate's probability of winning is higher when listed on the front / back page of the CIP than when listed on an inner page. Thus, if the causal effect of front / back page placement on the winning probability varies with the number of candidates, this suggests that the number of candidates affects the front / back page effect. Figure 6 shows that such a relationship is in fact observed in the dataset. Specifically, Figure 6 reports the difference in mean winning probabilities between candidates listed on the front page and those listed on inner pages for subsamples stratified by the number of candidates. The difference is larger in elections with 60 or more candidates than in elections with fewer than 50 candidates or with 50--59 candidates. Because candidates' page placements in the CIP are randomly determined, the difference in mean winning probabilities between front-page and inner-page candidates is the causal effect of front-page placement on the winning probability. Figure 6 therefore shows that the front page effect differs with the number of candidates.
\begin{landscape}
    \begin{table}[H]
        \centering
        \begin{threeparttable}
            \caption{Summary statistics for candidate-level variables}
            \begin{tabular}{@{}llccccc@{}}
                \toprule
                &  & Observations & Mean & SD & Min & Max \\
                \midrule
                Winning dummy &  & 8,450 & 0.750 & 0.433 & 0 & 1\\
                Page placement & Front page & 8,450 & 0.220 & 0.414 & 0 & 1\\
                & Back page & 8,450 & 0.158 & 0.365 & 0 & 1\\
                & Inner pages & 8,450 & 0.622 & 0.485 & 0 & 1\\
                Age &  & 8,450 & 50.609 & 11.178 & 25 & 91\\
                Gender & Male & 8,450 & 0.776 & 0.417 & 0 & 1\\
                & Female & 8,450 & 0.224 & 0.417 & 0 & 1\\
                Candidate status & Incumbent & 8,450 & 0.619 & 0.486 & 0 & 1 \\
                & Newcomer / Former & 8,450 & 0.381 & 0.486 & 0 & 1\\
                Party affiliation & Liberal Democratic Party & 8,450 & 0.278 & 0.448 & 0 & 1\\
                & Komeito & 8,450 & 0.153 & 0.360 & 0 & 1\\
                & Japanese Communist Party & 8,450 & 0.110 & 0.313 & 0 & 1\\
                & Minor parties & 8,450 & 0.236 & 0.425 & 0 & 1\\
                & Independents & 8,450 & 0.222 & 0.416 & 0 & 1\\
                \bottomrule
            \end{tabular}
        \end{threeparttable}
    \end{table}
\end{landscape}
\par
\begin{figure}[htbp]
    \centering
    \begin{minipage}{\textwidth}
        \centering
        \caption{Correlation between the number of candidates and the front page effect}
        \includegraphics[width=0.62\textwidth]{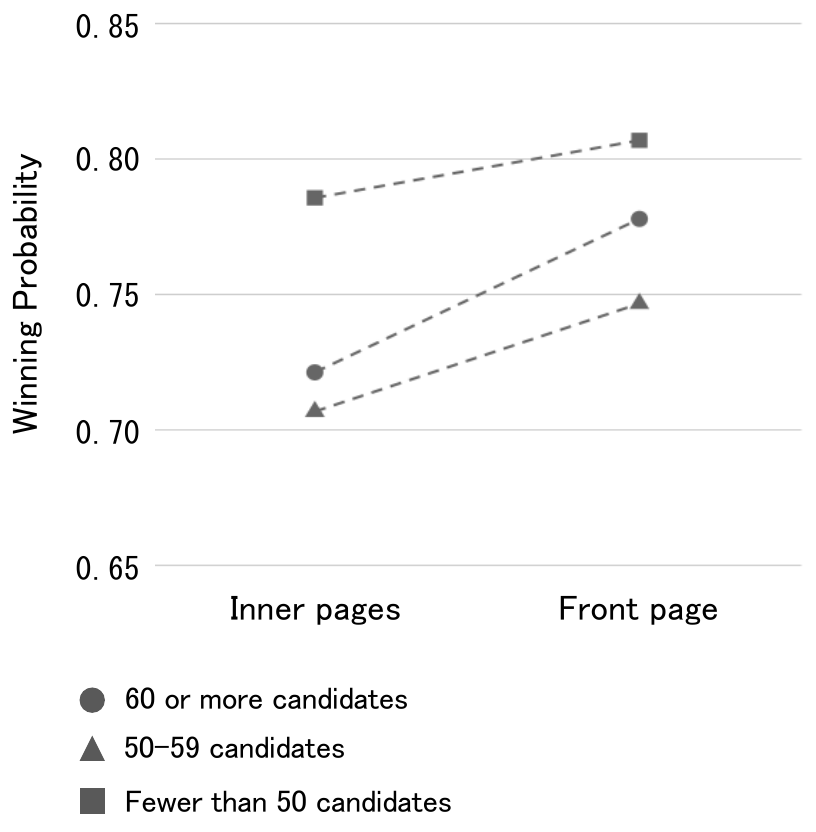}
    \end{minipage}
\end{figure}
\par

\section{Estimation Method}
    
    \subsection{Causal Moderation: Parallel Estimation Framework of Bansak (2021)}
        \textcite{bansak2021estimating} proposes a framework for identifying causal interaction effects between a randomized treatment with a small number of categories and a non-randomized moderator. Let $Y_i^P(t,s)$ denote the potential outcome corresponding to a binary treatment $T_i=t$ and moderator $S_i=s$, each taking values $0$ and $1$. The estimand of interest, the average treatment moderation effect ($\ATME$), can be written as follows:
\begin{equation*}
    \ATME \equiv \E\left[\left\{Y_i^P (1, 1) - Y_i^P (0, 1)\right\} - \left\{Y_i^P (1, 0) - Y_i^P (0, 0)\right\}\right].
\end{equation*}
That is, the $\ATME$ measures the extent to which the change in $Y$ induced by changing $T$ from $0$ to $1$ changes when $S$ changes from $0$ to $1$; equivalently, it measures the extent to which $S$ changes the effect of $T$ on $Y$.\par
Bansak shows that, under three standard assumptions (stable unit treatment value assumption: SUTVA, common support, and complete randomization of the treatment) and the two assumptions in \eqref{independent_T} and \eqref{conditional_independent_S}, $\delta^{PE}$ in \eqref{parallel_estimation_equation} equals the $\ATME$.
\begin{alignat}{1}
    &\left(Y_i^P (1, 1), Y_i^P (0, 1), Y_i^P (1, 0), Y_i^P (0, 0), S_i, \boldsymbol{X}_i^P\right) \indep T_i, \label{independent_T}\\
    &\text{where}\ \boldsymbol{X}_i^P\ \text{are}\ \text{covariates}. \notag
\end{alignat}
\begin{equation}
    \left(Y_i^P (1, 1), Y_i^P (0, 1), Y_i^P (1, 0), Y_i^P (0, 0)\right) \indep S_i \mid \boldsymbol{X}_i^P. \label{conditional_independent_S}
\end{equation}
\begin{alignat}{1}
    \gamma_t &= \E_{\boldsymbol{X}^P \mid T = t}\left[\E\left[Y_i^P (t, 1) \mid T_i = t, S_i = 1, \boldsymbol{X}_i^P\right] - \E\left[Y_i^P (t, 0) \mid T_i = t, S_i = 0, \boldsymbol{X}_i^P\right]\right], \notag\\
    \delta^{PE} &= \gamma_1 - \gamma_0. \label{parallel_estimation_equation}
\end{alignat}\par
In other words, the sample is split into the $T=1$ and $T=0$ strata; within each stratum, the difference in $Y$ between $S=1$ and $S=0$, conditional on $\boldsymbol{X}^P$, is estimated; and the difference between these two within-stratum differences identifies the causal interaction effect. For the number of candidates, which is the moderator in this paper, however, it is difficult to find control variables that fully satisfy \eqref{conditional_independent_S}. We therefore construct an identification strategy by applying the FDC.

    \subsection{Front-door Criterion of Pearl (1995)}
        \textcite{pearl1995causal} proposes a framework that, for a directed acyclic graph (DAG) such as Figure 7, identifies the causal effect of $H$ on $Y$ by decomposing it into the causal effect of $H$ on $M$ and the causal effect of $M$ on $Y$.
\begin{figure}[htbp]
    \centering
    \begin{minipage}{\textwidth}
        \centering
        \caption{DAG for the FDC}
        \includegraphics[width=0.65\textwidth]{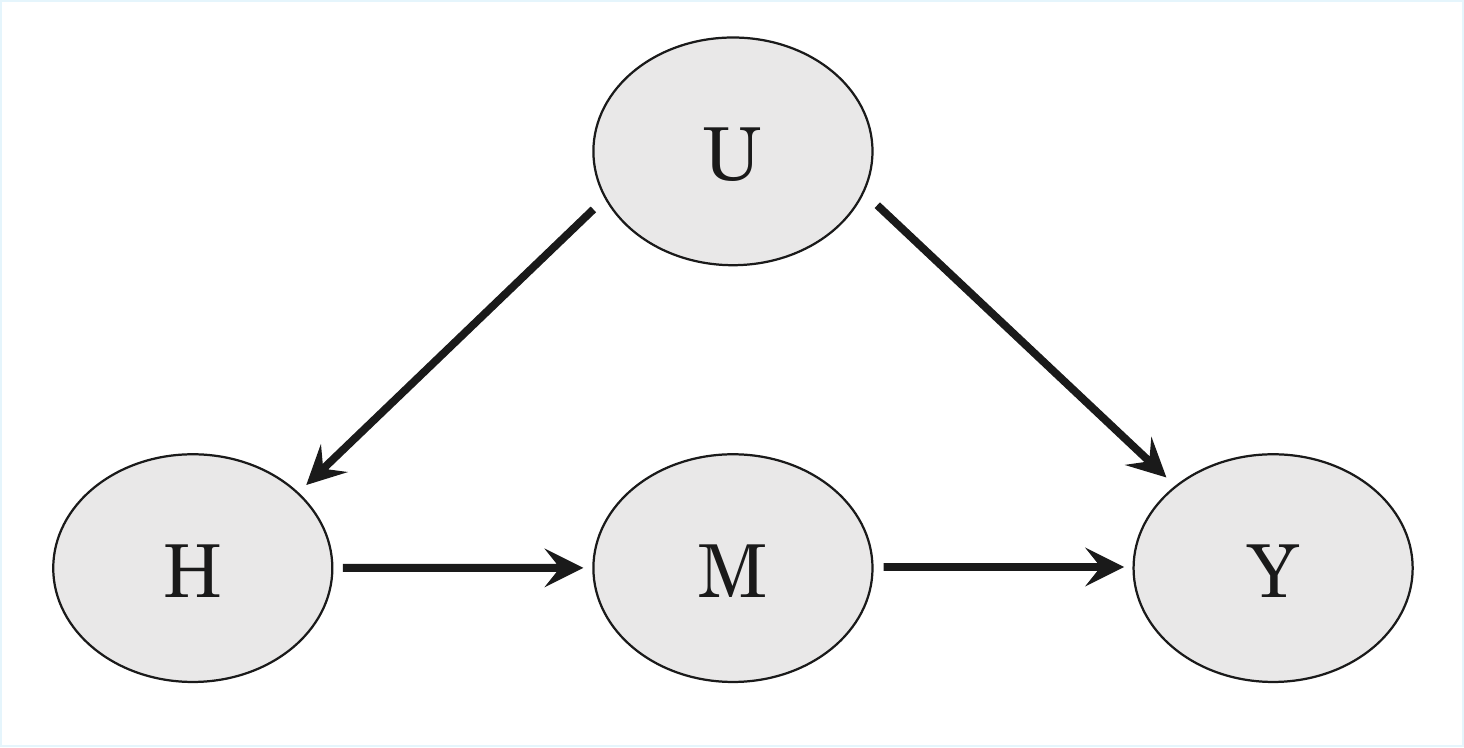}
    \end{minipage}
\end{figure}
\par
Building on this DAG-based argument, \textcite{bellemare2024paper} let $Y_i^F(h,m)$ denote the potential outcome corresponding to a binary variable $H_i=h$ and mediator $M_i=m$, each taking values $0$ and $1$, and let $M_i(h)$ denote the potential mediator corresponding to $H_i=h$. They organize the assumptions required for identification into the three conditions in \eqref{no_direct_paths}, \eqref{exogenous_H_to_M}, and \eqref{H_condition_M_independence}.
\begin{equation}
    Y_i^F (1, m) = Y_i^F (0, m),\ \forall m. \label{no_direct_paths}
\end{equation}
\begin{equation}
    (M_i (1), M_i (0)) \indep H_i. \label{exogenous_H_to_M}
\end{equation}
\begin{equation}
    \E\left[Y_i^F (h, m) \mid H_i = h, M_i = m\right] = \E\left[Y_i^F (h, m) \mid H_i = h\right],\ \forall m,h. \label{H_condition_M_independence}
\end{equation}\par
In this paper, when $H$ is the number of candidates and $M$ is the total number of CIP pages, confounders that would violate \eqref{exogenous_H_to_M} and \eqref{H_condition_M_independence} can be eliminated by controlling for municipality dummies $\boldsymbol{X}^F$ (details are provided below). In addition, \textcite{bellemare2024paper} show that, even when \eqref{no_direct_paths} does not hold, the path through which $H$ affects $Y$ via $M$ can still be identified.\footnote{In that case, however, the identified causal effect may not coincide with the total effect of $H$ on $Y$.} Under linearity, the estimation equations for identifying only the mediated path while controlling for $\boldsymbol{X}^F$ so that \eqref{exogenous_H_to_M} and \eqref{H_condition_M_independence} hold are as follows:
\begin{alignat}{1}
    \theta_0 &= \E_{\boldsymbol{X}^F} \left[\E\left[M_i (1) \mid H_i = 1, \boldsymbol{X}^F\right] - \E\left[M_i (0) \mid H_i = 0, \boldsymbol{X}^F\right]\right], \notag\\
    \theta_1 &= \E_{\boldsymbol{X}^F} \left[\E\left[Y_i^F (h, 1) \mid H_i = h, M_i = 1, \boldsymbol{X}^F\right] - \E\left[Y_i^F (h, 0) \mid H_i = h, M_i = 0, \boldsymbol{X}^F\right]\right], \notag\\
    \gamma^{FD} &= \theta_0 \times \theta_1. \label{front-door_estimation_equation}
\end{alignat}
That is, the causal mediation effect is estimated by taking the product of the difference in $M$ between $H=1$ and $H=0$, conditional on $\boldsymbol{X}^F$, and the difference in $Y$ between $M=1$ and $M=0$, conditional on $\boldsymbol{X}^F$ and $H$.

    \subsection{An Application of Front-door Criterion for Causal Moderation}
        Building on Sections 4.1 and 4.2, this subsection combines the PEF and the FDC to identify the causal moderation effect. Let $Y(t,h,m)$ denote the potential outcome corresponding to $T=t$, $H=h$, and $M=m$. Substituting \eqref{front-door_estimation_equation} from Section 4.2 into \eqref{parallel_estimation_equation} from Section 4.1 with $S=H$ and $\boldsymbol{X}^P=\boldsymbol{X}^F$, and treating $H$ and $M$ as continuous variables, yields the estimation equations in \eqref{joint_estimation_equation}.
\begin{alignat}{1}
    \theta_{0t} &= \E_{H, \boldsymbol{X}^F \mid T = t} \left[\E\left[\frac{\partial M_i (h)}{\partial h} \middle| T_i = t, H_i = h, \boldsymbol{X}^F\right]\right], \notag\\
    \theta_{1t} &= \E_{H, M, \boldsymbol{X}^F \mid T = t}\left[\E\left[\frac{\partial Y_i (t, h, m)}{\partial m} \middle| T_i = t, H_i = h, M_i = m, \boldsymbol{X}^F\right]\right], \notag\\
    \gamma_t^{FD} &= \theta_{0t} \times \theta_{1t}, \notag\\
    \delta^{PE \times FD} &= \gamma_1^{FD} - \gamma_0^{FD} \label{joint_estimation_equation}.
\end{alignat}
That is, the sample is split into the $T=1$ and $T=0$ strata. Within each stratum, we first estimate the marginal effect of $H$ on $M$ conditional on $\boldsymbol{X}^F$ and the marginal effect of $M$ on $Y$ conditional on $\boldsymbol{X}^F$ and $H$, and take their product. We then take the difference between the resulting estimator in the $T=1$ stratum and that in the $T=0$ stratum.
\begin{Proposition}
    Under the standard assumptions (SUTVA, common support, complete randomization of the treatment, and linearity) and Assumptions 1--3, the estimator in \eqref{joint_estimation_equation} identifies the average mediated treatment moderation effect ($\AMTME$) in \eqref{AMTME_estimand}.
\end{Proposition}
\begin{Assumption}
    \begin{equation*}
        \left(\{Y_i (t, h, m)\}_{t \in \mathscr{T}, h \in \mathscr{H}, m \in \mathscr{M}}, H_i, \{M_i (h)\}_{h \in \mathscr{H}}, \boldsymbol{X}_i^F\right) \indep T_i,
    \end{equation*}
    where $\mathscr{T}, \mathscr{H}, \mathscr{M}$ are the supports of $T$, $H$, and $M$.
\end{Assumption}
\begin{Assumption}
    \begin{equation*}
        \{Y_i (t, h, m)\}_{t \in \mathscr{T}, h \in \mathscr{H}, m \in \mathscr{M}} \indep M_i \mid H_i, \boldsymbol{X}_i^F.
    \end{equation*}
\end{Assumption}
\begin{Assumption}
    \begin{equation*}
        \{M_i (h)\}_{h \in \mathscr{H}} \indep H_i \mid \boldsymbol{X}_i^F.
    \end{equation*}
\end{Assumption}
\begin{Estimand}
    \begin{equation}
        \AMTME \equiv \E_{H, \boldsymbol{X}^F} \left[\frac{\partial \nu (h, \boldsymbol{x})}{\partial h} \left\{\frac{\partial \mu_1 (h, m, \boldsymbol{x})}{\partial m} - \frac{\partial \mu_0 (h, m, \boldsymbol{x})}{\partial m}\right\}\right], \vspace{-18pt} \label{AMTME_estimand}
    \end{equation}
    \begin{alignat*}{1}
        \text{where}\ &\nu (h, \boldsymbol{x}) = \E \left[M_i (h) \mid \boldsymbol{X}_i^F = \boldsymbol{x}\right],\\
        &\mu_t (h, m, \boldsymbol{x}) = \E\left[Y_i (t, h, m) \mid H_i = h, \boldsymbol{X}_i^F = \boldsymbol{x}\right].
    \end{alignat*}
\end{Estimand}

\section{The Change in Front / Back Page Effect from the Number of Candidates}
    
    \subsection{Estimation}
        Section 4 shows that the estimand of interest in this paper can be identified by applying the PEF and the FDC. On this basis, we assume a linear data generating process (DGP) and estimate the following models:
\begin{alignat}{1}
    \# \pages_{wy} &= \theta_{0t} \# \candidates_{wy} + \mu_w + v_{wy}, \label{FD_first}\\
    \win_{iwy} &= \theta_{1t} \# \pages_{wy} + \zeta_t \# \candidates_{wy} + \eta_w + \boldsymbol{x \beta}_t + \varepsilon_{iwy}, \label{FD_second}\\
    \gamma_t^{FD} &= \theta_{0t} \times \theta_{1t}, \label{FD_multiply}\\
    \delta^{PE \times FD} &= \gamma_1^{FD} - \gamma_0^{FD}. \label{PE_diff}
\end{alignat}
Here, $i$ indexes candidates, $w$ municipalities, and $y$ election years. $\# \pages$ denotes the total number of CIP pages, $\# \candidates$ the number of candidates, $win$ a dummy equal to $1$ if the candidate wins, and $\mu_w$ and $\eta_w$ municipality fixed effects. $\boldsymbol{x}$ is a vector of candidate characteristics (a male dummy, age, an incumbent dummy, and multiple dummies for party affiliation), which is included to improve estimation efficiency. In addition, because \eqref{FD_first}, \eqref{FD_second}, and \eqref{FD_multiply} are estimated after splitting the sample into candidates listed on the front / back / inner page(s), each coefficient carries a $t$ subscript indicating the corresponding category.\footnote{Under Assumption 1, the population coefficients $\theta_{00}$ and $\theta_{01}$ coincide. However, in a finite sample, the corresponding estimates need not be identical. The subscript $t$ is therefore retained to make explicit that the coefficients estimated separately across strata may differ.} $t=1$ denotes the front / back page and $t=0$ denotes the inner pages. We first split the sample into the front, back, and inner strata and estimate \eqref{FD_first} and \eqref{FD_second} by ordinary least squares (OLS) within each stratum. We then compute $\hat{\gamma}_t^{FD}$ as the product of the coefficient estimates $\hat{\theta}_{0t}^{OLS}$ and $\hat{\theta}_{1t}^{OLS}$, following \eqref{FD_multiply}. Finally, following \eqref{PE_diff}, we calculate the difference $\hat{\delta}^{PE \times FD}$ between the front-page (back-page) estimate $\hat{\gamma}_1^{FD}$ and the inner-page estimate $\hat{\gamma}_0^{FD}$ to identify the extent to which the number of candidates changes the front / back page effect.\par
As shown in Section 3.1, candidates' positions in the CIP are determined by lottery. Therefore, the setting in this paper satisfies Assumption 1 required for applying the PEF. As also shown in Section 3.1, the total number of CIP pages is mechanically determined by the number of candidates according to each municipality's ordinance. Moreover, the dataset is constructed so that each municipality's sample period contains no change in the ordinance specifying the pamphlet format. Thus, conditional on the number of candidates and municipality fixed effects, the total number of CIP pages is not affected by unobserved factors such as electoral salience and is exogenous. The setting therefore satisfies Assumptions 2 and 3 required for the FDC.

    \subsection{Results}
        Tables 3--6 report the estimation results. The estimates show that one additional candidate increases the winning probability by $0.148$ percentage points among candidates listed on the front page ($p<0.1$), decreases it by $0.021$ percentage points among candidates listed on the back page ($p>0.1$), and decreases it by $0.174$ percentage points among candidates listed on inner pages ($p<0.05$). The difference between the front-page and inner-page estimates is $0.322$ percentage points ($p<0.05$), while the difference between the back-page and inner-page estimates is $0.153$ percentage points ($p>0.1$).
\begin{table}[H]
    \centering
    \begin{threeparttable}
        \caption{First-stage estimation results based on the FDC}
        \begin{tabular}{@{}lccc@{}}
            \toprule
            & \multicolumn{3}{c}{Total number of pages}\\
            \cmidrule{2-4}
            & Front & Back & Inner\\
            \midrule
            Number of candidates & \hspace{15pt} 0.100*** & \hspace{15pt} 0.096*** & \hspace{15pt} 0.131***\\
            & (0.027) & (0.022) & (0.035)\\
            \midrule
            Municipality FE & Y & Y & Y\\
            Sample size & 1,858 & 1,334 & 5,257\\
            \bottomrule
        \end{tabular}
        \begin{tablenotes}[flushleft]
            \footnotesize
            \setlength{\labelsep}{0pt}
            \item \textbf{Note}: ***, **, and * indicate statistical significance at the 1\percent, 5\percent, and 10\percent{} levels, respectively. Standard errors clustered at the municipality level are reported in parentheses.
        \end{tablenotes}
    \end{threeparttable}
\end{table}
\par
\begin{table}[H]
    \centering
    \begin{threeparttable}
        \caption{Second-stage estimation results based on the FDC}
        \begin{tabular}{@{}lccc@{}}
            \toprule
            & \multicolumn{3}{c}{Winning dummy}\\
            \cmidrule{2-4}
            & Front & Back & Inner\\
            \midrule
            Total number of pages & \hspace{15pt} 1.472*** & -0.222 \hspace{2pt}& \hspace{11pt} -1.333***\\
            & (0.460) & (0.841) & (0.358)\\
            \midrule
            Number of candidates & Y & Y & Y\\
            Municipality FE & Y & Y & Y\\
            Covariates & Y & Y & Y\\
            Sample size & 1,858 & 1,334 & 5,257\\
            \bottomrule
        \end{tabular}
        \begin{tablenotes}[flushleft]
            \footnotesize
            \setlength{\labelsep}{0pt}
            \item \textbf{Note}: Coefficient estimates and standard errors are reported after multiplying by 100. ***, **, and * indicate statistical significance at the 1\percent, 5\percent, and 10\percent{} levels, respectively. Standard errors clustered at the municipality level are reported in parentheses.
        \end{tablenotes}
    \end{threeparttable}
\end{table}
\par
\begin{table}[H]
    \centering
    \begin{threeparttable}
        \caption{Product estimation results based on the FDC}
        \begin{tabular}{@{}lccc@{}}
            \toprule
            & \multicolumn{3}{c}{Winning dummy}\\
            \cmidrule{2-4}
            & Front & Back & Inner\\
            \midrule
            Number of candidates \hspace{3pt} & \hspace{5pt} 0.148* & -0.021 \hspace{2pt}& \hspace{6pt} -0.174**\\
            & (0.080) & (0.080) & (0.086)\\
            \bottomrule
        \end{tabular}
        \begin{tablenotes}[flushleft]
            \footnotesize
            \setlength{\labelsep}{0pt}
            \item \textbf{Note}: Coefficient estimates and standard errors are reported after multiplying by 100. ***, **, and * indicate statistical significance at the 1\percent, 5\percent, and 10\percent{} levels, respectively. Standard errors reported in parentheses are computed using the cluster bootstrap method with clustering at the municipality level.
        \end{tablenotes}
    \end{threeparttable}
\end{table}
\par
\begin{table}[H]
    \centering
    \begin{threeparttable}
        \caption{Difference estimation results based on the PEF}
        \begin{tabular}{@{}lcc@{}}
            \toprule
            & Front page effect & Back page effect\\
            \midrule
            Number of candidates & \hspace{10pt} 0.322** & 0.153\\
            & (0.138) & (0.147)\\
            \bottomrule
        \end{tabular}
        \begin{tablenotes}[flushleft]
            \footnotesize
            \setlength{\labelsep}{0pt}
            \item \textbf{Note}: Coefficient estimates and standard errors are reported after multiplying by 100. ***, **, and * indicate statistical significance at the 1\percent, 5\percent, and 10\percent{} levels, respectively. Standard errors reported in parentheses are computed using the cluster bootstrap method with clustering at the municipality level.
        \end{tablenotes}
    \end{threeparttable}
\end{table}
\par
Because the baseline front page effect in the sample is $4.023$ percentage points,\footnote{This value is obtained by estimating by OLS estimation of a linear regression model with the winning dummy as the outcome and only the front-page and back-page dummies as explanatory variables.} the increase of $0.322$ percentage points is substantial. We therefore conclude that the number of candidates increases the front page effect, whereas the results do not support that it increases the back page effect. These results indicate that an increase in the number of candidates changes voting decision making even when time constraints are relatively weak, but that this change is limited to greater concentration on information encountered early in the process.

    \subsection{Limitations}
        First, the analysis uses election data in which the number of candidates ranges from 25 to 82 and estimates the effect of the number of candidates on the front / back page effect under a linearity assumption. However, some studies find inverted U-shaped relationships between the number of options and choice probabilities or satisfaction, peaking at three options \parencite{taagepera2014turnout} or ten options \parencite{shah2007buying}. Therefore, the results of this paper may fail to capture nonlinear relationships between the number of candidates and the front / back page effect, particularly differences in the relationship when the number of candidates is in the single digits.\par
Second, because the analysis uses the FDC, it identifies only the part of the effect of the number of candidates on the front / back page effect that operates through the total number of CIP pages. Any effects operating through other paths cannot be captured by the present analysis. Therefore, we cannot strongly rule out the possibility that, if effects through other paths are large enough to overturn the present estimates, the conclusion regarding the overall effect of the number of candidates on the voting decision-making process would change.\par
Third, the dataset constructed for this analysis contains 165 of the 487 municipal assembly elections held in Tokyo from 1987 to 2023. This reflects data-collection constraints, such as the absence of preserved election records for earlier years, as well as identification constraints, such as excluding periods containing changes to ordinances specifying the pamphlet format. If tendencies for election records to be preserved or for pamphlet formats to be revised are correlated with the number of candidates or the front / back page effect, the results may not represent the overall pattern for the target period and area.\par
Finally, because of data limitations, this paper cannot examine heterogeneity in the effect of the number of candidates on the front / back page effect. The analysis also relies on the CIP, a relatively uncommon medium.\footnote{\textcite{haruhara2026unintended} identify the Voters' Pamphlet in the United States, the Mayoral Election Address Booklet in the United Kingdom, and the profession de foi in France as examples of similar official media providing candidate information.} The results should therefore be extrapolated with caution to elections in other regions or periods, or to elections in which the primary medium for providing candidate information differs substantially.

\section{Conclusions and Policy Implications}
    Using data from municipal assembly elections held from 1987 to 2023, this paper applies the FDC to analyze how the number of candidates affects the magnitude of the front / back page effect. The results show that one additional candidate increases the front page effect by $0.322$ percentage points through changes in the total number of CIP pages. This increase amounts to $8.0$\percent{} of the baseline front page effect.\par
The results show that an increase in the number of candidates changes ordinary voter behavior---information acquisition during the campaign period and the subsequent choice of whom to vote for. In particular, the increase in the front page effect indicates that the voting decision-making process becomes simpler on average. Moreover, because the FDC identifies a specific mediator, the results make clear that this phenomenon arises from increases in the costs of information acquisition and comparison. Thus, within the elections in this paper's sample, an increase in the number of candidates increases the divergence between voters' actual choices and their accurate preferences.\par
This conclusion implies that there may be benefits to policies that directly reduce the number of candidates by increasing the cost of candidacy, as well as to policies that reduce the number of candidates each voter must compare by dividing electoral districts. It therefore suggests that electoral institutions should be adjusted toward an optimal number of candidates by weighing these benefits against losses from a narrower range of candidate attributes and risks such as gerrymandering. Alternatively, if voters' preferences can be predicted from their characteristics and past choices and information can be presented in an order consistent with those preferences, the divergence between actual choices and accurate preferences may be reduced.\par
The conclusions may also apply to choices among goods that resemble candidates in elections. For example, candidates share features of experience and credence goods with professional services such as those provided by lawyers, accountants, and consultants. Because consumers have difficulty predicting ex post outcomes from their prior information for experience and credence goods, an increase in the number of options may be especially likely to induce them to abandon information acquisition or evaluation. Therefore, for experience or credence goods for which the number of available options can be regarded as excessive, institutional arrangements may be needed either to reduce the number of options faced by each consumer or to enable information to be presented in a manner aligned with individual consumer preferences.

\clearpage
\printbibliography

\clearpage
\appendix
\section*{Appendix: Proof of Proposition 1}
    \setcounter{Proposition}{0}
\setcounter{Assumption}{0}
\subsection*{A.1. Setup and the estimand}

    Let $T_i \in \{0, 1\}$ denote the randomized treatment, let $H_i \in \mathscr{H} \subseteq \mathbb{R}$ denote the moderator, and let $M_i \in \mathscr{M} \subseteq \mathbb{R}$ denote the mediator. Let $M_i(h)$ be the potential mediator under $H_i = h$, and let $Y_i(t, h, m)$ be the potential outcome under $T_i = t, H_i = h$, and $M_i = m$. Let $\boldsymbol{X}_i^F$ denote the covariates used for the front-door adjustment.\par
    Although $H$ and $M$ are count-valued in the empirical application, the estimation treats them as continuous variables.  Accordingly, the estimand below is defined on the continuously extended linear causal-response surfaces imposed in Section 4.  Under linearity, the derivative with respect to $H$ is also the effect associated with a one-unit increase in $H$.\par
    Define the mean causal response of the mediator by
    \begin{equation}
        \nu(h, \boldsymbol{x}) \equiv \E \left[M_i(h) \mid \boldsymbol{X}_i^F = \boldsymbol{x} \right], \label{eq:nu}
    \end{equation}
    and, for $t \in \{0, 1\}$, define the mean causal response of the outcome to the mediator, within the $H_i = h$ stratum, by
    \begin{equation}
        \mu_t(h, m, \boldsymbol{x}) \equiv \E \left[Y_i(t, h, m) \mid H_i = h, \boldsymbol{X}_i^F = \boldsymbol{x}\right]. \label{eq:mu}
    \end{equation}
    The conditional causal effect of $T$ at $(h, m, \boldsymbol{x})$ is
    \begin{equation*}
        \tau(h, m, \boldsymbol{x}) \equiv \mu_1(h, m, \boldsymbol{x}) - \mu_0(h, m, \boldsymbol{x}).
    \end{equation*}\par
    The total marginal change in this treatment effect as $H$ changes along the mediator response $m = \nu(h, \boldsymbol{x})$ can be decomposed by the chain rule:
    \begin{alignat}{1}
        \frac{d}{dh} \tau \left(h, \nu(h, \boldsymbol{x}), \boldsymbol{x}\right) &= \underbrace{\frac{\partial \tau(h, m, \boldsymbol{x})}{\partial h} \bigg|_{m = \nu(h, \boldsymbol{x})}}_{\text{moderation through paths not mediated by }M} \notag\\
        &\quad + \underbrace{\frac{\partial \tau(h, m, \boldsymbol{x})}{\partial m} \bigg|_{m = \nu(h, \boldsymbol{x})}\frac{\partial \nu(h, \boldsymbol{x})}{\partial h}}_{\text{moderation through }H \rightarrow M \rightarrow Y}. \label{eq:chain}
    \end{alignat}\par
    The parameter of interest is the average of the second component only.
    \begin{equation}
        \boxed{
            \begin{aligned}
                \AMTME & \equiv \E_{H, \boldsymbol{X}^F} \left[ \frac{\partial \nu(h, \boldsymbol{x})}{\partial h}\left\{\frac{\partial \mu_1 (h, m, \boldsymbol{x})}{\partial m} - \frac{\partial \mu_0 (h, m, \boldsymbol{x})}{\partial m}\right\}\right]\bigg|_{\substack{h = H_i\\m = \nu(H_i, \boldsymbol{X}_i^F)}}.
            \end{aligned}
        } \label{eq:amtme}
    \end{equation}\par
    Thus, the average mediated treatment moderation effect ($\AMTME$) measures the marginal change in the causal effect of $T$ induced by a marginal change in $H$ \emph{only through the mediator $M$}.  In particular, the direct $H \rightarrow Y$ component,
    \begin{equation*}
        \E_{H, \boldsymbol{X}^F}\left[\frac{\partial \mu_1 (h, m, \boldsymbol{x})}{\partial h} - \frac{\partial \mu_0 (h, m, \boldsymbol{x})}{\partial h}\right],
    \end{equation*}
    is not part of $\AMTME$.

\subsection*{A.2. Assumptions}
    
    For formal completeness, Assumption 1 is stated with the full collection of potential mediators included.  This is the potential-outcome counterpart of complete randomization of $T$.
    \begin{Assumption}[Randomization of $T$]
        \label{ass:T}
        \[\left(\{Y_i(t, h, m)\}_{t \in \mathscr{T}, h \in \mathscr{H}, m \in \mathscr{M}}, \{M_i(h)\}_{h \in \mathscr{H}}, H_i, \boldsymbol{X}_i^F\right) \indep T_i.\]
    \end{Assumption}
    \begin{Assumption}[Conditional exogeneity of $M$ for the outcome]
        \label{ass:M_Y}
        \[\{Y_i(t, h, m)\}_{t \in \mathscr{T}, h \in \mathscr{H}, m \in \mathscr{M}} \indep M_i \mid H_i, \boldsymbol{X}_i^F.\]
    \end{Assumption}
    \begin{Assumption}[Conditional exogeneity of $H$ for the mediator]
        \label{ass:H_M}
        \[\{M_i(h)\}_{h \in \mathscr{H}} \indep H_i \mid \boldsymbol{X}_i^F.\]
    \end{Assumption}\par
    In addition, impose SUTVA / consistency, the relevant common-support conditions, and linearity. For the present proof, linearity means that the conditional mean causal-response surfaces can be written as
    \begin{align}
        \nu(h,\boldsymbol{x}) &= a_M(\boldsymbol{x}) + \theta_0 h, \label{eq:linearM}\\
        \mu_t(h, m, \boldsymbol{x}) &= a_t(\boldsymbol{x}) + \zeta_t h + \theta_{1t} m, \qquad t \in \{0, 1\}, \label{eq:linearY}
    \end{align}
    with differentiation and expectation interchangeable. Importantly, $\zeta_t$ is not restricted to zero; hence a direct $H \rightarrow Y$ path is allowed.

\subsection*{A.3. Proposition and proof}

    For notational completeness, write the two arm-specific causal marginal effects used in the parallel estimation framework as
    \begin{alignat}{1}
        \theta_{0t} &\equiv \E_{H, \boldsymbol{X}^F \mid T = t}\left[\E \left[\frac{\partial M_i(h)}{\partial h} \middle| T_i = t, H_i = h, \boldsymbol{X}_i^F \right]\right], \label{eq:theta0t}\\
        \theta_{1t} &\equiv \E_{H, M, \boldsymbol{X}^F \mid T = t} \left[\E \left[\frac{\partial Y_i(t, h, m)}{\partial m} \middle| T_i = t, H_i = h, M_i = m, \boldsymbol{X}_i^F\right]\right], \label{eq:theta1t}\\
        \gamma_t^{FD} &\equiv \theta_{0t} \theta_{1t},
        \qquad
        \delta^{PE\times FD} \equiv \gamma_1^{FD}-\gamma_0^{FD}. \label{eq:delta}
    \end{alignat}
    The outer expectation in \eqref{eq:theta1t} includes $M$ so that no free $m$ remains. Under the linearity assumption, this additional averaging is numerically immaterial because the marginal effect with respect to $m$ is constant.
    \begin{Proposition}[Identification of the $\AMTME$]
        \label{prop:1}
        Under SUTVA/consistency, common support, complete randomization of $T$, linearity, and Assumptions \ref{ass:T}--\ref{ass:H_M}, 
        \begin{equation*}
            \delta^{PE\times FD} = \AMTME.
        \end{equation*}
        Equivalently,
        \begin{equation*}
            \delta^{PE \times FD} = \theta_0 \left(\theta_{11} - \theta_{10} \right),
        \end{equation*}
        where $\theta_0$ is the causal marginal effect of $H$ on $M$, and $\theta_{1t}$ is the causal marginal effect of $M$ on $Y$ under $T=t$.
    \end{Proposition}
    
    \begin{proof}
        We proceed in four steps.
        
        \paragraph{Step 1: Identification of the causal marginal effect of $H$ on $M$.}

            By consistency, for units with $H_i = h, M_i = M_i(h)$. Therefore,
            \begin{alignat}{1}
                \E \left[M_i \mid H_i = h, \boldsymbol{X}_i^F = \boldsymbol{x}\right] &= \E \left[M_i(h) \mid H_i = h, \boldsymbol{X}_i^F = \boldsymbol{x}\right] \notag\\
                &= \E \left[M_i(h) \mid \boldsymbol{X}_i^F = \boldsymbol{x}\right], \label{eq:first-id}
            \end{alignat}
            where the second equality follows from Assumption \ref{ass:H_M}. Hence
            \begin{equation*}
                \E[M_i \mid H_i = h, \boldsymbol{X}_i^F = \boldsymbol{x}] = \nu(h, \boldsymbol{x}).
            \end{equation*}
            Differentiating with respect to $h$ and applying \eqref{eq:linearM},
            \begin{equation*}
                \frac{\partial}{\partial h}\E \left[M_i \mid H_i = h, \boldsymbol{X}_i^F = \boldsymbol{x}\right] = \frac{\partial \nu(h, \boldsymbol{x})}{\partial h} = \theta_0.
            \end{equation*}\par
            Assumption \ref{ass:T} implies that conditioning on $T_i = t$ does not change this causal response. Consequently, for both $t = 0$ and $t = 1$,
            \begin{equation}
                \theta_{0t} = \theta_0. \label{eq:theta0equal}
            \end{equation}

        \paragraph{Step 2: Identification of the causal marginal effect of $M$ on $Y$ within each $T$ arm.}

            By Assumption \ref{ass:M_Y},
            \begin{alignat}{1}
                &\E\left[Y_i(t, h, m) \mid H_i = h, M_i = m, \boldsymbol{X}_i^F = \boldsymbol{x}\right] \notag\\
                &\qquad = \E\left[Y_i(t, h, m) \mid H_i = h, \boldsymbol{X}_i^F = \boldsymbol{x}\right] = \mu_t(h, m, \boldsymbol{x}). \label{eq:second-exog}
            \end{alignat}
            Assumption \ref{ass:T} allows $T_i = t$ to be added to the conditioning set without changing the potential-outcome mean. Consistency then gives
            \begin{equation}
                \E\left[Y_i \mid T_i = t, H_i = h, M_i = m, \boldsymbol{X}_i^F = \boldsymbol{x}\right] = \mu_t(h, m, \boldsymbol{x}). \label{eq:second-id}
            \end{equation}
            Differentiating \eqref{eq:second-id} with respect to $m$ and applying \eqref{eq:linearY},
            \begin{equation}
                \frac{\partial}{\partial m}\E\left[Y_i \mid T_i = t, H_i = h, M_i = m, \boldsymbol{X}_i^F = \boldsymbol{x}\right] = \frac{\partial \mu_t(h, m, \boldsymbol{x})}{\partial m} = \theta_{1t}. \label{eq:second-slope}
            \end{equation}
            Thus the second component of the front-door product in the $T = t$ arm identifies the causal marginal effect of $M$ on $Y$ under $T = t$.

        \paragraph{Step 3: Identification of the mediated effect within each $T$ arm.}

            From \eqref{eq:linearM} and \eqref{eq:linearY},
            \begin{equation*}
                \frac{\partial \nu(h, \boldsymbol{x})}{\partial h} = \theta_0,
                \qquad
                \frac{\partial \mu_t (h, m, \boldsymbol{x})}{\partial m} = \theta_{1t}.
            \end{equation*}
            Hence the mediated marginal effect of $H$ on $Y$ within arm $t$ is
            \begin{alignat}{1}
                \gamma_t^{M} &\equiv \E_{H, \boldsymbol{X}^F}\left[\frac{\partial \nu(h, \boldsymbol{x})}{\partial h}\frac{\partial \mu_t (h, m, \boldsymbol{x})}{\partial m}\right] \notag\\
                &= \theta_0 \theta_{1t}. \label{eq:arm-mediated}
            \end{alignat}
            By \eqref{eq:theta0equal}, this is exactly
            \begin{equation*}
                \gamma_t^{M} = \theta_{0t} \theta_{1t} = \gamma_t^{FD}.
            \end{equation*}

        \paragraph{Step 4: Difference across the two randomized $T$ arms.}

            Taking the difference between $t=1$ and $t=0$,
            \begin{alignat}{1}
                \delta^{PE \times FD} &= \gamma_1^{FD} - \gamma_0^{FD} \notag\\
                &= \theta_0 \theta_{11} - \theta_0 \theta_{10} \notag\\
                &= \theta_0 \left(\theta_{11} - \theta_{10}\right). \label{eq:product-diff}
            \end{alignat}
            On the other hand, substituting \eqref{eq:linearM} and \eqref{eq:linearY} into the definition of the AMTME in \eqref{eq:amtme} yields
            \begin{alignat}{1}
                \AMTME &= \E_{H, \boldsymbol{X}^F}\left[\theta_0 \left(\theta_{11} - \theta_{10}\right)\right] \notag\\
                &= \theta_0 \left(\theta_{11} - \theta_{10}\right).
            \end{alignat}
            Therefore,
            \begin{equation*}
                \boxed{
                    \delta^{PE \times FD} = \AMTME
                }.
            \end{equation*}\par
            Finally, note that $\zeta_t$, which represents the direct $H \rightarrow Y$ component in \eqref{eq:linearY}, does not enter any step of \eqref{eq:product-diff}. Therefore no restriction of the form $Y_i(t,h,m)=Y_i(t,h',m)$ is required for identification of the $\AMTME$. Such a restriction would be required only if the front-door product were interpreted as the \emph{total} effect of $H$ on $Y$.
    \end{proof}
\end{document}